\documentclass[11pt]{paper}
\pdfoutput=1
\usepackage{graphicx}
\usepackage{color}
\usepackage{amsmath}
\usepackage{bbm}
\usepackage{amssymb}
\usepackage{amsthm}
\usepackage{xspace}
\usepackage{algorithm}
\usepackage{hyperref}
\usepackage{fullpage}
\usepackage{tabularx}
\usepackage{multirow}
\usepackage[noend]{algpseudocode}

\algnewcommand{\LineComment}[1]{\State \(\triangleright\) #1}
\newtheorem{observation}{Observation}

\newtheorem{definition}{Definition}
\newtheorem{Lemma}{Lemma}
\newtheorem{theorem}{Theorem}

\newif\ifshowcomments
\showcommentstrue

\ifshowcomments
\newcommand{\mynote}[2]{\fbox{\bfseries\sffamily\scriptsize{#1}}
 {\small$\blacktriangleright$\textsf{\emph{#2}}$\blacktriangleleft$}}
\else
\newcommand{\mynote}[2]{}
\fi

\begin{document}

\title{Synchronous Byzantine Lattice Agreement in ${\cal O}(\log (f))$ Rounds}
\author{Giuseppe Antonio Di Luna$^{\,\,\dag}$
\and
Emmanuelle Anceaume$^{*}$
\and
Silvia Bonomi$^{\,\,\dag}$
\and
Leonardo Querzoni$^{\,\,\dag}$}
\institution{${\dag}$: Sapienza University of Rome - {\{diluna,bonomi,querzoni\}@diag.uniroma1.it}
\and
${*}$: CNRS, Univ Rennes, Inria, IRISA - emmanuelle.anceaume@irisa.fr}

\maketitle

\hrule
\begin{abstract}
In the Lattice Agreement (LA) problem, originally proposed by Attiya et al. \cite{Attiya:1995}, a set of processes has to decide on a chain of a lattice. More precisely, each correct process proposes an element $e$ of a certain join-semi lattice $L$ and it has to decide on a value that contains $e$. Moreover, any pair $p_i,p_j$ of correct processes has to decide two values $dec_i$ and $dec_j$ that are comparable (e.g., $dec_i \leq dec_j$ or $dec_j < dec_i$).  LA has been studied for its practical applications, as example it can be used to implement a snapshot objects \cite{Attiya:1995}  or a replicated state machine with commutative operations \cite{Faleiro:2012}.  
Interestingly, the study of the Byzantine Lattice Agreement started only recently, and it has been mainly devoted to asynchronous systems.  The synchronous case has been object of a recent pre-print \cite{Zheng:aa} where Zheng et al. propose an algorithm terminating in  ${\cal O}(\sqrt f)$ rounds and tolerating $f < \lceil n/3 \rceil$ Byzantine processes. 

In this paper we present new contributions for the synchronous case. We investigate the problem in the usual message passing model for a system of $n$ processes with distinct unique IDs. We first  prove that, when only authenticated channels are available, the problem cannot be solved if $f=n/3$ or more processes are Byzantine. We then propose a novel algorithm that works in a synchronous system model with signatures (i.e., the {\em authenticated message} model), tolerates up to $f$ byzantine failures (where $f<n/3$) and that terminates in ${\cal O}(\log f)$ rounds.  We discuss how to remove authenticated messages at the price of algorithm resiliency ($f < n/4$). Finally, we present a transformer that converts any synchronous LA algorithm to an algorithm for synchronous Generalised Lattice Agreement. 
\end{abstract}
\hrule

\section{Introduction}

Fault-tolerance is a key research topic in distributed computing \cite{10.5555/2821576}: reliable algorithms have been deeply investigated in classic message passing \cite{doi:10.1002/0471478210} and in newer model of computations \cite{popprot,icdcn,michail,doty}. A special mention has to be given to algorithms for distributed agreement \cite{105555,6812814} . They represent a cornerstone of todays cloud-based services. In particular, practical and efficient implementations of distributed consensus, transformed in the last 10 years the Internet from a large computers network in a world-scale service platform. Despite its fundamental role, real implementations of distributed consensus are still confined to small scales or tightly controlled environments; the well known FLP result, in fact, makes distributed consensus impossible to solve deterministically in asynchronous settings, where communication latencies cannot be bounded. To cope with this limit, practical systems trade off consistency criteria (allowing weaker \emph{agreement} properties) with liveness (guaranteeing run termination only in long-enough grace periods where the systems ``behaves'' like a synchronous one).

For such a reason, agreement properties weaker than consensus proved to be extremely effective for the implementation of a broad family of distributed applications, since they can be used in systems where consensus cannot be solved, or they can be faster than consensus algorithms circumventing time-complexity lower bounds.

\noindent {\bf Lattice Agreement} \textemdash~ In this paper we investigate an agreement problem that is weaker than consensus: the Lattice Agreement (LA) problem. In LA, introduced by Attiya  et al.~\cite{Attiya:1995}, each process $p_i$ has  an input value $x_i$ drawn from the join semilattice  and must decide an output value $y_i$, such that \emph{(i)} $y_i$ is the join of $x_i$  and some set of input values and \emph{(ii)} all output values are comparable to each other in the lattice, that is they are all located on a single chain in the lattice (see Figure \ref{slattice}). LA describes situations in which processes need to obtain some knowledge on the global execution of the system, for example a global photography of the system. 
In particular Attiya  et al.~\cite{Attiya:1995} have shown that in the asynchronous shared memory computational model, implementing a snapshot object is equivalent to solving the Lattice Agreement problem.

Faleiro  et al.~\cite{Faleiro:2012} have shown that in a message passing system a majority of correct processes and reliable communication channels are sufficient to solve LA in  asynchronous systems, proposing a Replicated State Machine with commutative updates built on top of a generalized variant of their LA algorithm. Generalized Lattice Agreement (GLA) is a version of LA where processes propose and decide on a, possibly infinite, sequence of values. A recent solution of Skrzypczak  et al.~\cite{Skrzypczak} considerably improves Faleiro's construction in terms of memory consumption, at the expense of liveness. 

The restriction of having only commutative updates is justified by the possibility of developing faster algorithms. It is well known \cite{ doi:10.1002/0471478210} that consensus cannot be solved in synchronous systems in less than ${\cal O}(f)$ rounds, even when only crash failures are considered, on the other hand it has been shown \cite{Garg:2018} that, when crash failures are considered, Lattice Agreement can be solved in ${\cal O}(\log f)$ rounds \footnote{Actually, it can be solved faster than $\log(f)$ on specific lattices, the ones having height less than $\log f$ \cite{Garg:2018}, however in this paper we are interested only in worst case performances.}, in ${\cal O}(\log f)$ message delays in asynchronous \cite{Garg:2018b}. 

\noindent {\bf Byzantine Failures}  \textemdash~ All these papers consider process failures, but assume that such failures are non-malicious. More recently, some works started proposing LA algorithms that tolerate Byzantine failures. 
The first one has been by Nowak and Rybicki~\cite{Nowak:2019}: they introduced LA with Byzantine failures and proposed a definition of Byzantine LA in which decisions of correct processes are not allowed to contain values proposed by Byzantine processes. Using such definition, authors have shown that the number of correct processes needed to solve LA depends on the structure of the lattice, and they have shown a LA algorithm for specific kinds of lattices. 

Successively, Di Luna et al. \cite{ipdps} proposed a less restrictive definitions of Byzantine LA, in which correct processes can decide also values proposed by Byzantine. Authors have shown that, when their definition is used, LA can be solved for any possible lattice when $f < \frac{n}{3}$: they proposed a solution for Byzantine LA in asynchronous systems that terminates in $\mathcal O(f)$ rounds; the same paper also proposed a Generalized version of the algorithm and built on top of it a Replicated State Machine that executes commutative operations. In this paper we adopt the Byzantine LA definition from \cite{ipdps}, since it allows to circumvent the impossibility of \cite{Nowak:2019} and it is usable in many practical scenario.  The same definition has been also used by Zheng and Garg~\cite{Zheng:aa}, where they show that LA can be solved in synchronous systems with $\mathcal O(\sqrt f)$ rounds also in presence of Byzantine failures. 

\begin{figure}
  \centering
    \includegraphics[width=0.5\textwidth]{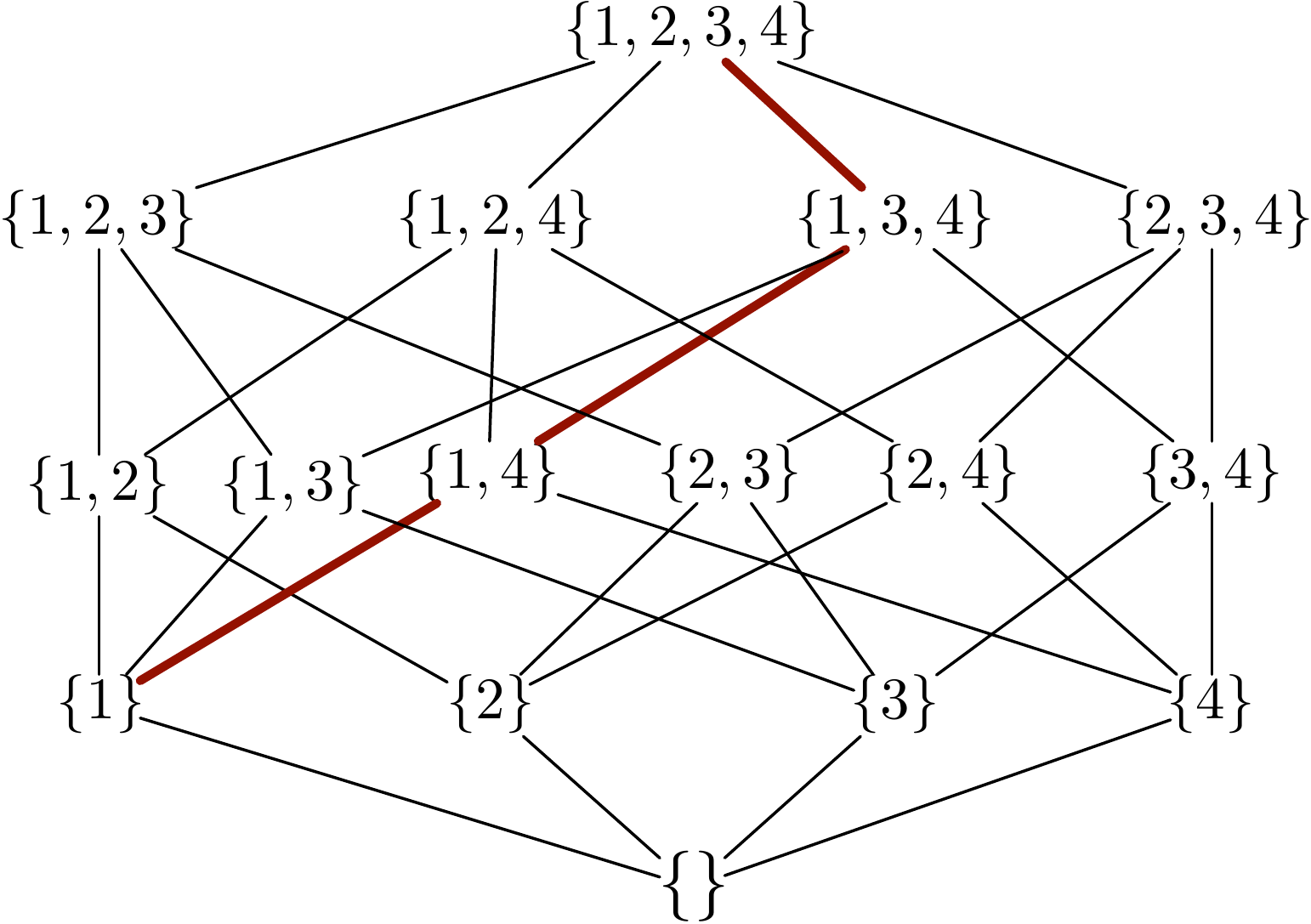}
      \caption{Hasse diagram of the semilattice induced over the power set of $\{1,2,3,4\}$ using the union operation as join. Given any two elements $e,e'$ of the lattice if $e < e'$, then $e$ appears lower in the diagram than $e'$ and there is an ``upward" path, going from lower points to upper points, connecting $e$ to $e'$ (e.g., $\{1\} \leq \{1,3,4\}$, but $\{2\} \not\leq \{3\}$). Any two elements $e,e'$ of the semilattice have a join $e \oplus e'= e \cup e'$ and $e \oplus e' \geq e,e'$ (e.g., $\{1\} \oplus \{2,3\}=\{1,2,3\}$ ). The red edges indicate a possible chain (i.e., sequence of increasing values) that contains the output values. 
	  }\label{slattice}
\end{figure}

\noindent {\bf Contributions}  \textemdash~In this paper we present new contributions for the Byzantine LA problem in synchronous settings.
Our first results is for systems with only authenticated channels (i.e., signatures are not available), in such systems we show that Byzantine LA on arbitrary lattices cannot be solved, in synchronous systems, with $f=\lceil n/3 \rceil$ or more faulty processes (Section \ref{authchannel:impossibility}).
Interestingly, such proof shows that the algorithm of Zheng and Garg~\cite{Zheng:aa} is tight in the number of tolerable failures. 
On the positive side we show algorithms that solve LA and Generalized LA, with and without signatures, having better running time that the state-of-the-art. 

Looking at the model with signatures, we show a novel algorithm for LA that works in a synchronous system model, tolerates up to $f$ byzantine failures (where $f<\lceil n/3 \rceil$) and that terminates in ${\cal O}(\log f)$ rounds. The algorithm improves over the $LA_\beta$ algorithm from Garg at al.~\cite{Garg:2018} by using a similar construction, but adding tolerance to Byzantine failures. We make use of a modified Gradecast algorithm that allows processes to prove that a message has been seen by all correct processes in the system. (Sections \ref{signatures}-\ref{sec:general})

We conclude our investigation on LA by briefly discussing how to remove signatures and make our construction work only with authenticated channels trading-off part of its resiliency: we are able to tolerate only $f<\lceil n/4 \rceil$ failures (Section \ref{sec:nosig}).
  
In the last part of the manuscript, we devote our attention to  Generalized Lattice Agreement (Section \ref{bzgc:const}). Specifically, we show a transformer that using as building block a generic LA algorithm creates a Generalized Lattice Agreement algorithm. At the best of our knowledge this is the first time GLA is investigated in synchronous systems. 

\subsection{Related Work} 
\label{sec:related_work}

\renewcommand{\arraystretch}{1.2}
 \begin{table*}[t]
 \footnotesize
 \caption{Comparison of algorithms for Byzantine Lattice Agreement when the lattice has an height that is greater than $f$. \label{table:comparison}}
\begin{tabularx}{\textwidth}{ | c | c | X | X | X | X | }
 \hline
 	{\em Model} & {\em Assumption} & {\em Paper} &  {\em Resiliency}  & {\em Time} & {\em Messages}  \\
 \hline
 \hline

 	\multirow{4}{*}{({\sf SYNC})} & \multirow{2}{*}{Auth. Links (No Sign.)}  & This paper &  $f<\lceil n/4 \rceil$ & ${\cal O}(\log(f))$ & ${\cal O}(n^2 \log(f))$ \\
	   \cline{3-6}
 &  & Zheng et al. \cite{Zheng:aa} &  $f <\lceil n/3 \rceil$ &${\cal O}(\sqrt{f})$  &  ${\cal O}(n^2 \sqrt{f})$ \\

 	 \cline{2-6}
	& \multirow{1}{*}{Auth. Messages (Signatures)}  & This paper &$f< \lceil n/3 \rceil$ & ${\cal O}(\log(f))$ &  ${\cal O}(n^2 \log(f))$ \\
 	 \hline
	  \hline
	 \multirow{2}{*}{({\sf ASYNC})}& Auth. Links (No Sign.)&    Di Luna et al. \cite{ipdps} & $f< \lceil n/3 \rceil$ &  $   {\cal O}(f) $ & ${\cal O}(n^2)$ \\
           		 \cline{2-6}
	 & Auth. Messages (Signatures) &  Di Luna et al. \cite{ipdps} & $f< \lceil n/3 \rceil$ &  $   {\cal O}(f) $ & $ {\cal O}(f \cdot n)$\\

 \hline
\end{tabularx}

\end{table*}
A map of the relevant related work is illustrated in Table \ref{table:comparison}, where we compare Byzantine fault-tolerant algorithms using the definition of LA from \cite{ipdps}. 
Zheng and Garg~\cite{Zheng:aa} showed that Byzantine LA can be solved in synchronous settings with $\mathcal{O}(\sqrt{f})$ rounds. The algorithm they propose makes use of a modified version of the Gradecast~\cite{Ben-or:2010} algorithm as a building block. Furthermore, correct processes are asked to keep track of a lattice of safe values among which final values will be decided. This approach guarantees that Byzantine processes cannot pollute correct process decisions with values that are not safe, i.e. values that are known by every correct process. When signatures are available, the algorithm we propose matches the same resiliency $f<\lceil n/3 \rceil$ (that we show being optimal) while having a faster running time. When we remove signatures our construction has a worse resiliency $f<\lceil n/4 \rceil$ but it keeps a faster running time. 
Di Luna et al.~\cite{ipdps} proposed a solution for Byzantine LA in asynchronous settings, also providing an algorithms for its generalized variant. Their algorithm bases its correctness on an initial phase were values to be proposed are broadcasted to build a safe set of values in the lattice from which the final decided values will be chosen. Byzantine processes that try to propose arbitrary values that are not contained in the safe set will see their message simply ignored.

\section{System Model, Notation, and Preliminaries}

We use the usual message passing models with unique identifiers (IDs). There is a set $\Pi$ of $n$ processes with unique IDs in $\{1,\ldots,n\}$ connected by a complete communication graph.  The system is synchronous, and the execution of the algorithm can be divided in discrete finite time units called rounds. In each round a process is able to send messages to its neighbours ({\em send phase}), and receive all messages sent to it at the beginning of the round ({\em receive phase}). Processes in $\Pi$ are partitioned in two sets $F$ and $C$. Processes in $C$ are correct, they faithfully follow the distributed protocol. Processes in $F$ are Byzantine, they arbitrarily deviate from the protocol. As usual when Byzantine failures are considered, we assume that the communication channels are authenticated by mean of Message Authentication Codes (MAC).
The {\em authenticated channels} are the only assumption used in Section~\ref{authchannel:impossibility}. In Section~\ref{signatures} we assume that there is a public key infrastructure that allows processes to cryptographically sign messages, that can be lately verified by other processes. This model has {\em authenticated messages}. Byzantine processes are polynomially bounded and cannot forge signatures of correct processes.
For an easier presentation we explain our algorithms for the case of $n=3f+1$, where $f=|F|$, however they can be easily adapted for any other $n> 3f+1$. 
\paragraph{Notation.}
With $\epsilon$ we indicate the empty string. Given a string $G$, with $|G|$ we indicate the length of the string ($|\epsilon|=0$), with $G[j]$ and $0 \leq j < |G|$ we indicate the character of string $G$ in position $j$. With $G[k:l]$ (given $0 \leq k \leq l \leq |G|$), we indicate the substring of $G$ between position $k$ and $l$.
As an example given $G=ssms$, we have $G[0]=s$ and $G[0:1]=ss$.
 Given two strings $a$ and $b$ with $a \cdot b$ we indicate the string obtained by concatenating $b$ after $a$. 

\subsection{The Byzantine Lattice Agreement Problem}  \label{definition:la}
  Each process $p_i \in C$ starts with an initial input value $pro_i \in E$ with $E \subseteq V$ (set $E$ is a set of allowed proposal values). 
Values in $V$ form a join semi-lattice $L=(V,\oplus)$ for some commutative join operation $\oplus$: for each $u,v \in V$ we have $u \leq v$ if and only if $v= u \oplus v$. 
 Given  $V'=\{v_1,v_2,\ldots,v_k\} \subseteq V$ we have $\bigoplus V'=v_1 \oplus v_2 \oplus \ldots \oplus v_k$.

  The task that processes in $C$ want to solve is the one of  Lattice Agreement, and it is formalised by the following properties:
  \begin{itemize}
\item {\bf Liveness:} Each process $p_i \in C$ eventually outputs a decision value $dec_i \in V$;
\item {\bf Stability:} Each process $p_i \in C$  outputs a unique decision value $dec_i \in V$;
\item {\bf Comparability:} Given any two pairs $p_i,p_j \in C$ we have that either $dec_i \leq dec_j$ or $dec_j \leq dec_i$;
\item {\bf Inclusivity:} Given any correct process $p_i \in C$ we have that $pro_i \leq dec_i $;
\item {\bf Non-Triviality:}  Given any correct process $p_i \in C$ we have that $dec_i \leq \bigoplus (X \cup B)$, where $X$ is the set of proposed values of all correct processes ($X:\{pro_i | \text{ with } p_i \in C \}$), and $B \subseteq E$ is $|B| \leq f$.

  \end{itemize}
  
 {\bf  Lattice definitions.} Given two distinct elements $u,v \in V$ ve say that there exists a path between $u$ and $v$ if they are comparable; a path of length $k$ between $u$ and $v$ is a sequence of $k+1$ distinct elements $(e_0,e_2,\ldots,e_k)$ of the lattice such that
 $e_0=u \leq e_1 \leq e_2 \leq \ldots \leq e_{k-1} \leq e_k=v$. As an example the path between $\{1,2,3\}$ and $\{1\}$ in the lattice of Figure \ref{slattice} has length 2.
 We say that a  $v \in V$ is minimal if it does not exists $u \in V$, with $u \neq v$, such that $u \oplus v= v$ (i.e., it does not exists an $u \leq v$).  As in \cite{Zheng:aa} we define the 
 height of an element $v$ in a lattice $(V, \oplus)$ has the length of the longest path from any minimal element to $v$ in the lattice (as an example the Lattice in Figure \ref{slattice} has height $4$).  A sub-lattice of $(V,\oplus)$ is a subset $U$ of $V$ closed with respect the join operation, the definition of height for a sub-lattice does not change.   
  
   {\bf  Preliminaries.}  In the rest of the paper we will assume that $L$ is a semi-lattice over sets ($V$ is a set of sets) and $\oplus$ is the set union operation. This is not restrictive, it is well known~\cite{Nation} that any join semi-lattice is isomorphic to a semi-lattice of sets with  set union  as join operation. An important  lattice is the one on the power set of the first $\{1,\ldots,n\}$ natural numbers with the union as join operation (see Figure \ref{slattice}), we will use as shorthand for such lattice the notation $L_{n}$. Such lattice is interesting for our purpose, we will show that an algorithm solving lattice agreement exclusively on such lattice (the GAC of Section \ref{signatures}) can be used as building block to solve lattice agreement on an arbitrary lattice (Section \ref{sec:general}). 
   When $L_{n}$ is considered, the height of an element $e$ is equal to its cardinality (i.e, $|e|$ cfr. Figure \ref{slattice}), given a sub lattice of $L_{n}$ its height is upper bounded by the difference between the minimum and maximum cardinality of its elements.

\section{Authenticated Channels and No Signature - Necessity of $3f+1$ processes in synchronous systems}  \label{authchannel:impossibility}
In the following we show that $3f+1$ processes are necessary when there are no signatures.

\begin{Lemma}\label{lemmabase}
It does not exist any algorithm solving Byzantine LA on arbitrary lattices in a synchronous system with $3$ processes when one is faulty. The impossibility holds even when relaxing the Non-Triviality allowing $|B| \leq k$ for any fixed $k$.  
\end{Lemma}
\begin{proof}
We first discuss the case of $k=1$. 
Let ${\cal A}$ be an algorithm solving LA with $3$ processes when one is faulty. Since ${\cal A}$ works on arbitrary lattices it should also work on the lattices induced by the union operation on the power set of the first $6$ natural numbers with $E=\{\{1\},\{2\}, \{3\}, \{4\}, \{5\}, \{6\}\}$.
Now let us consider the hexagonal system of Figure \ref{3fstart}. Such a system is constituted by $6$ processes $p_1,p_2,p_3,p_4,p_5,p_6$ with an edge between each $p_i,p_j$ such that $i=j \pm 1$ and one edge between $p_1$ and $p_6$.  
Each of the six processes has as input an unique value in $[1,6]$, just for simplicity process $p_i$ has input $\{i\}$. 
Note that even if ${\cal A}$ is an algorithm for three processes, it is possible to execute ${\cal A}$ on the hexagon, but its behaviour does not necessarily follows the LA specification.  

On the right of the figure we have $6$ triangles, each triangle is related to a corresponding edge in the hexagon. 
The relationship is such that the view of two neighbour processes in the hexagon is equal to the view of two processes in a triangle where the third process is a Byzantine simulating the behaviour of the other processes in the hexagon. 
As example: the view of processes $p_1,p_2$ in the hexagon  is the same that $p_1,p_2$ would have in triangle $t_{green}$, analogously the view of $p_6,p_1$ in the hexagon is the same view of $p_6,p_1$ in $t_{blue}$.   Note that ${\cal A}$ once executed on any of the triangle in the figure has to follow the  LA specification.  

A run of ${\cal A}$ on the hexagon in principle has an undefined behaviour. However we observe that a run of ${\cal A}$ on the hexagon eventually terminates on each process, this is because each process has a local view that is consistent with a system of 3 processes one of which is a Byzantine. Recall, that the local view of each process $p_i$ in the hexagon is exactly the same view that the process has in the two triangles on the right, and ${\cal A}$  being a correct algorithm when $3$ processes are considered the
algorithm will correctly terminate in each triangle in the right. 

Moreover, each process will output a decision value that must be the same that the process will output in the corresponding triangles. 

Let $dec_1,dec_2,dec_3,dec_4,dec_5,dec_6$ be the decisions of processes dictated by ${\cal A}$ (naturally we have $dec_i$ decision of $p_i$).   As reference see Figure \ref{3fend}. 
The triangles on the right impose a certain number of comparability relationships among these decisions. Recall, that each decision is a subset, not necessarily proper, of $[1,6]$, and that
the comparability in this setting is the relationship of inclusion. An example is triangle $t_{green}$ that imposes the comparability between $dec_1$ and $dec_2$, that is either $dec_1 \subset dec_2$ or $dec_2 \subseteq dec_1$.
In the following we use $ dec_i \leftrightarrow dec_j$ to indicate that $dec_i$ must be comparable with $dec_j$.  
Before continuing with the proof we give the following technical observation.
\begin{observation}
Consider a collection of $m$ sets $S_1,S_2,\ldots,S_m$ such that for each $S_i$ we have $i \in S_i$. If it holds that $S_j \leftrightarrow S_{i+1}$ for all $j \in [1,m-1]$ then there exists an $|S_k| \geq m$.  
 \end{observation}

Therefore, let us take w.l.o.g. $dec_1$, and let us walk in clockwise direction for 3 steps on the hexagon. On this walk we have:
$dec_1 \leftrightarrow dec_2 \leftrightarrow dec_3 \leftrightarrow dec_4$, by the inclusivity property we have for each $dec_i$ that $i \in dec_i$.
Therefore we can apply our technical lemma and state that one of the $dec_i$ has cardinality at least $4$ violating the non-triviality property of ${\cal A}$ on some triangle (recall that $k$ being equal to $1$ we have $|B| \leq 1$).
The generalised proof for an arbitrary $k$ follows the same reasoning using: a lattice on the power set of the first $3(k+1)$ natural numbers, $E=\{ \{1\},\ldots, \{3(k+1)\}\}$ and a $3(k+1)-$gon instead of an hexagon. We then walk on the $3(k+1)-$gon for $k+2$ steps instead of $3$. 
We will have a chain  $dec_1 \leftrightarrow dec_2 \leftrightarrow \ldots \leftrightarrow dec_{k+3}$ where by inclusivity of ${\cal A}$ on the corresponding triangle we have $i \in dec_i$, and where our 
observation shows that one of the decisions contains at least $k+3$ distinct elements of $E$ violating the non-triviality on some triangle. \end{proof}

\begin{theorem}\label{th:impossible}
It does not exist any algorithm solving Byzantine LA on arbitrary lattices in a synchronous system with $3f$ processes and $f$ faulty if $|B| \leq f$.  The same holds if $n < 3f$
\end{theorem}
\begin{proof}
Let ${\cal A}$ be an algorithm that solves $LA$ on a system of $3f$ processes.
We build an algorithm ${\cal A}_{sym}$ for three processes using ${\cal A}$.  In algorithm ${\cal A}_{sym}$ each process simulates $f$ processes of ${\cal A}$. Each of these simulated processes starts with proposal of the real process.
Once a simulated process decides the corresponding real process decides the same set.
It is immediate that ${\cal A}_{sym}$ is a correct lattice agreement algorithm with a relaxed non-triviality assumptions that includes at most $f$ values from a Byzantine. 
This violates Lemma \ref{lemmabase}.
When the number of failures is greater than $n/3$ the same argument holds.
\end{proof}

The proof of Theorem \ref{th:impossible} does not work in a system with authenticated messages (i.e., signatures), it is therefore unkown whether $3f+1$ processes are necessary also in this model.
Interestingly, in \cite{ipdps} it is shown that, when the system is asynchronous, $3f+1$ processes are necessary also when authenticated messages are available.

\begin{figure}
  \centering
    \includegraphics[width=0.8\textwidth]{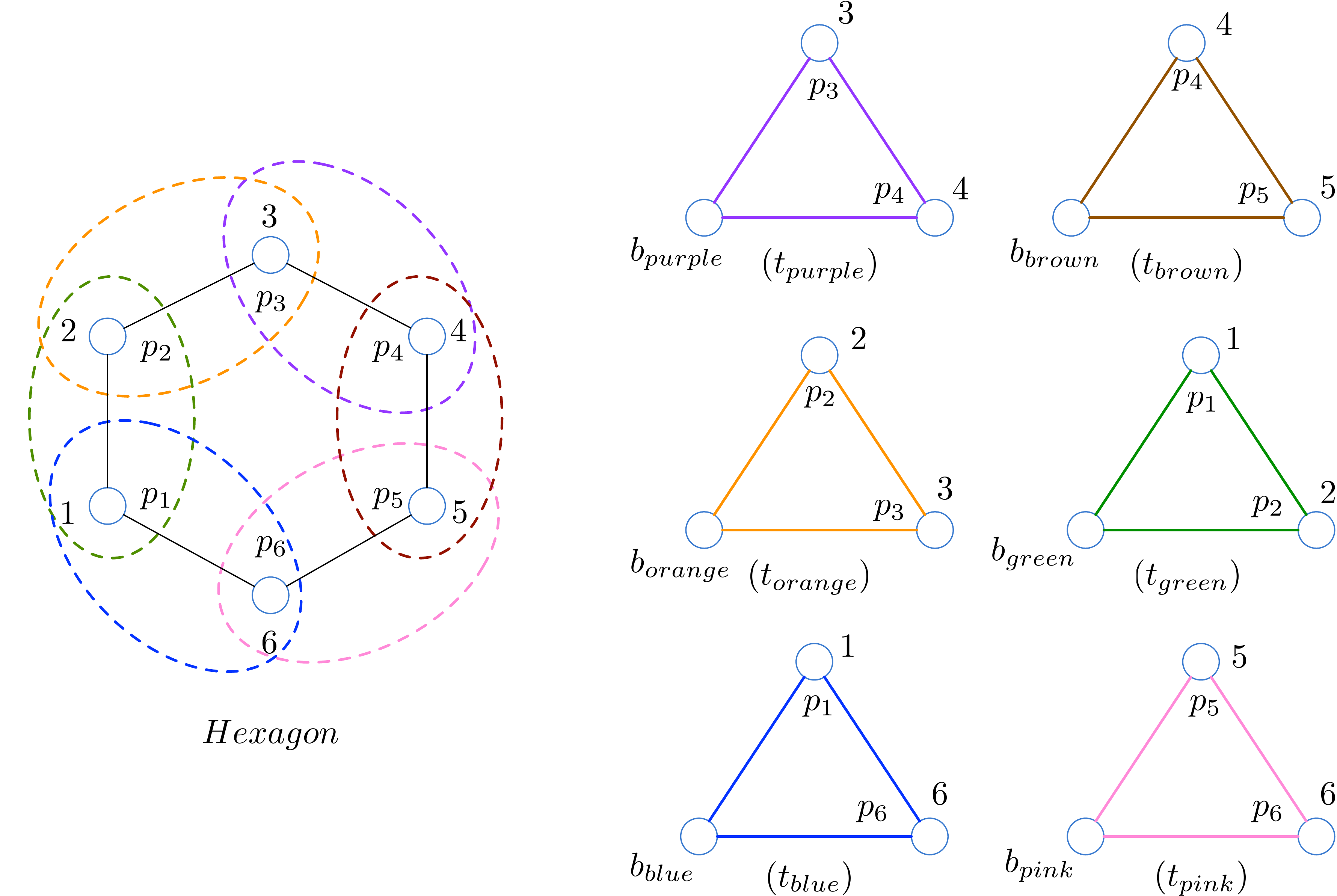}
      \caption{Starting configuration of the Hexagon. For each edge there is a corresponding triangle that is a legal starting configuration of ${\cal A}$ with one Byzantine. As an example, take $p_1,p_2$ in the hexagon. 
      They have the exactly same view of $p_1,p_2$ in $t_{green}$ where the other node is $b_{green}$ a Byzantine that simulates the behaviour of $p_3,p_4,p_5,p_6$ in the hexagon. \label{3fstart} }
\end{figure}

\begin{figure}
  \centering
    \includegraphics[width=0.8\textwidth]{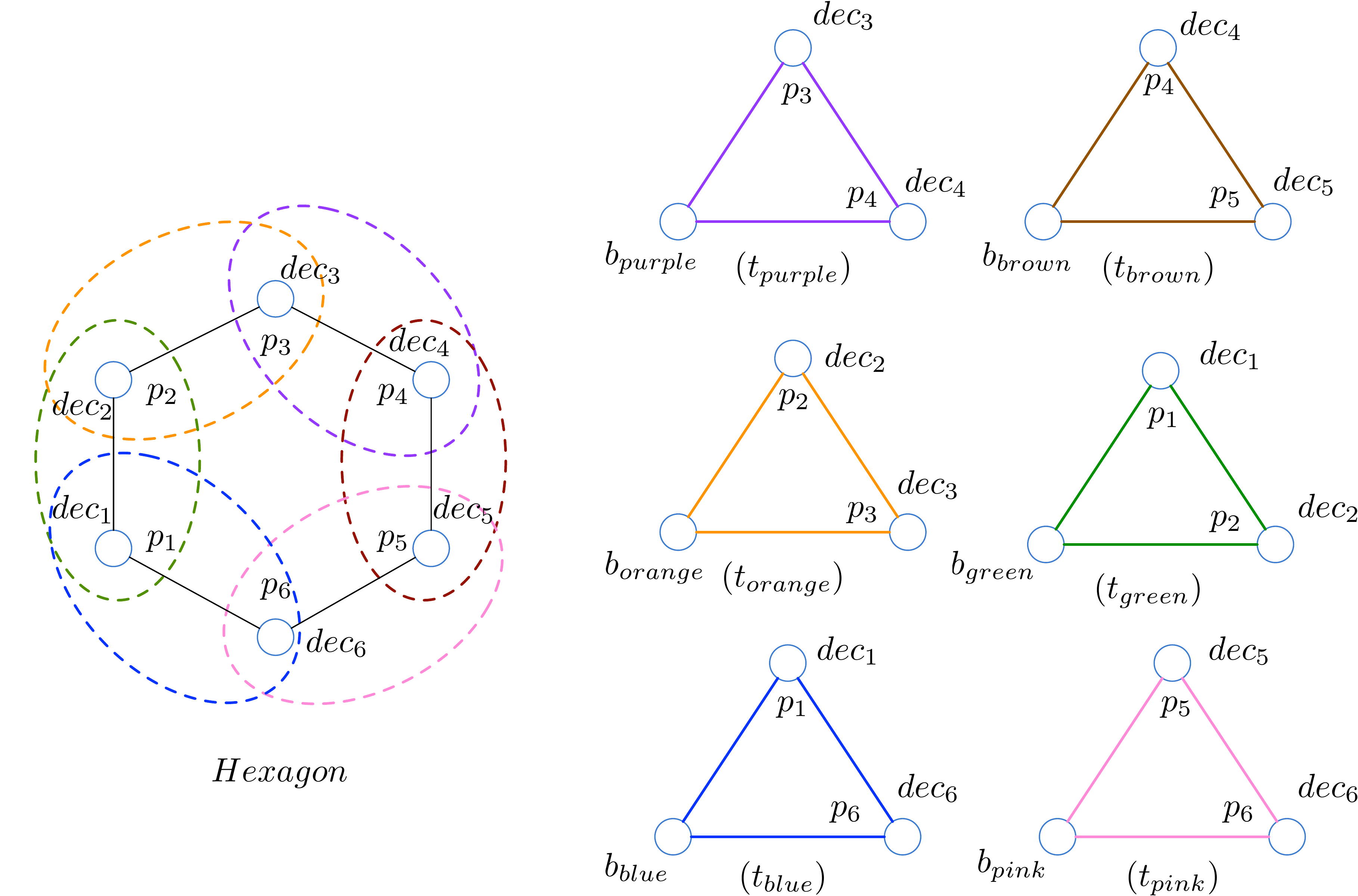}
      \caption{Ending configuration of the Hexagon.  \label{3fend} }
\end{figure}


\section{Algorithm for $L_{n}$: Grade And Classify (GAC)} \label{signatures}
 In this section we show an algorithm that works on $L_{n}$ (note that $n$ is also the number of processes). The algorithm terminates in ${\cal O}(\log(n))$ rounds.  We will then discuss in Section \ref{sec:general} how to use this algorithm so solve LA on arbitrary join semi-lattices, and how to adapt it to work in ${\cal O}(\log(f))$ rounds. 

\paragraph{High Level Description.} Our algorithm is based on the algorithmic framework of Zheng et al.~\cite{Garg:2018} adapted to tolerate Byzantine failures. As in the original, the algorithm works by continuously partitioning processes in \emph{masters} and \emph{slaves} sets. Partitioning is recursively operated in successive epochs.  Processes that have been assigned to the same partition in each epoch are a ``group''. We indicate a generic group $G$ at epoch $ep$ with the string $s \cdot \sigma$, where $\sigma$ is in $\{s,m\}^{ep-1}$. As an example, the string $ssm$ indicates the set of process that at the end of epochs $0$ and $1$ entered the group of slaves (string $ss$) and then, at the end of epoch $2$, entered in the masters group (string $ssm$).  When we write $p \in G$, we indicate that process $p$ belongs to the group of processes identified by string $G$.
 
The algorithm then enforces some properties on the partitions generated at an epoch $ep$ on a Group $G$:
\begin{itemize}
\item each master process in $G\cdot m$ adjusts its proposal to be a superset of each possible decision of a slave process in $G \cdot s$.
\item the height of a sub-lattice in which processes in $G\cdot m$ (or $G\cdot s$) are allowed to decide halves at each epoch. 
\end{itemize}
Thanks to the above properties each group becomes independent and has to solve the lattice agreement on a lattice that has half of the height of the original. 
A key concept for our algorithm is the one of {\em ``admissible value"}, a value is admissible, for a certain epoch, if it is ensured that it will not conflict with decisions of processes that have been elected as masters in a previous epoch. This is done by showing a cryptographical proof that such a value can be accepted by a slave since it is in the proposal value of each master the slave could conflict with. 
After ${\cal O}(\log(n))$  rounds the algorithm terminates (each group operates on a lattice
constituted by a single point). 

\subsection{The Provable Gradecast Primitive}\label{pgc:sec}

The algorithm makes use of the gradecast primitive introduced by Ben-Or et al. in ~\cite{Ben-or:2010}. Such primitive is similar to a broadcast, we have a sender process $p_i$ that sends a message $m$, each other process $p_j$,  after 3 rounds,  outputs a tuple $(p_i,m_j,c_j)$ where $c_j \in \{0,1,2\}$ is a score
of the correctness of $p_i$. The gradecast ensures the following properties:
\begin{itemize}
\item for any two correct processes $p_j,p_\ell$ if $c_j >0$ and $c_\ell >0$ than $m_j=m_\ell$.
\item for any two correct processes $p_j,p_\ell$ we have $|c_j - c_\ell| \leq 1$
\item if the gradecast sender is correct, than for any correct process $p_j$ we have $c_j=2$.  
\item for any correct process $p_j$ if $(p_i,m_j,0)$ then $m_j=\bot$.
\end{itemize}

Intuitively, if we let processes communicate by mean of the gradecast primitive we force Byzantines to send at most two different messages to the set $C$ of correct processes, and one of these messages has to be $\bot$. 
We modify the original gradecast to make it ``provable". In our version of the gradecast each correct process outputs a tuple composed by 4 objects $(p_i,m_j,c_i, S_{i,m_j})$ where $S_{i,m_j}$ is a special object that can either be $\bot$ or a {\em seen-all} proof. In case $S_{i,m_j}$ is different from $\bot$, then it is a cryptographic proof that can be shown to other processes and it implies that, any correct process $p \in C$ has seen a rank, for the gradecast of message $m_j$ from process $p_i$, that is at least $1$. Moreover, we have that if $c_i=2$ then $S_{i,m_j} \neq \bot$. 
Practically, the modification of the original gradecast are contained, and are limited to the second and third round of the algorithm. 
The original gradecast, with source $p_s$ works as follows: in the first round $p_s$ broadcasts a message $m$ to all processes; in the second round each correct process relays the message received by $p_s$ (it ignores messages from other processes); at the end of the second round a correct process selects the most frequent message received, and if such message was received by at least a quorum of $n-f$ processes then it relays the message at the beginning of the third round; at the end of the third round each correct selects the most frequent message received and it ranks it $2$ and delivers it if the message was received by at least $n-f$ processes; if it was received by at least $f+1$ it delivers it and ranks it $1$, otherwise it delivers $\bot$ with rank $0$.

 In our version, see Algorithm \ref{pgc:algorithm}, the relaying process signs the relay sent at the beginning of round 3 (see line \ref{pgc:sign}) 
and a process that sees a message with rank 2 collects the $n-f$ (i.e., $2f+1$ if $n=3f+1$) signed messages (see line \ref{pgc:proof}). These signed messages constitute a proof that the message has been seen by all: delivered by each correct with rank at least 1. Note that the algorithm is for a single instance and a single determined sender, however one can trivially run in parallel an instance for each possible sender in the system. 

We do not prove the properties discussed above for our version of gradecast, they immediately derives from the correctness of the original algorithm \cite{Ben-or:2010}. We do show the property introduced by the seen-all proof:


\begin{algorithm*}
\caption{{\em Provable GradeCast (PGC)}  - Algorithm for sender process $p_s$ and receiver process $p_i$} \label{pgc:algorithm}
\footnotesize

\begin{algorithmic}[1]

\State ${\bf SND~phase~ of~ round~1}$ 
\If{$p_i=p_s$}
\State {\sc Broadcast}$(m)$ \Comment{Executed  only by sender $p_s$}
\EndIf
\LineComment{Code executed by any correct process $p_i \in \Pi$}
\State ${\bf RCV~phase~ of~ round~1}$ 
\State $rcv= {\sc receive\_messages()}$
\State $m$=select one message from $p_s$ in $rcv$ 
\State ${\bf SND~phase~ of~ round~2}$ 
\State {\sc Broadcast}$(m)$
\State ${\bf  RCV~phase~ of~ round~2}$ 
\State $rcv= {\sc receive\_messages()}$
\State $m=$  select the message in $rcv$ received from the largest set of distinct sources. 
\State ${\bf SND~phase~ of~ round~3}$ 
\If{ $m$ has been received from at least $n-f$ sources}  \label{pgc:relay}
\State {\sc Broadcast}($m, {\sc sign}(m)$) \label{pgc:sign} \Comment{$p_i$ relays the most frequent message signed by himself}
\EndIf
\State ${\bf  RCV~phase~ of~ round~3}$ 
\State $rcv= {\sc receive\_messages()}$
\State remove from $rcv$ all messages that are not correctly signed
\State $m_i=$  select the message in $rcv$ received from the largest set of distinct sources.. 
\If{ $m_i$  has been received by at least  $n-f$ sources} \label{pgc:freqr2}
\State $c_i=2$
\State $S_{p_s,m_i}=$ select $n-f$ correctly signed messages in $rcv$ for message $m_i$  \label{pgc:proof}
\ElsIf{ $m_i$  has been received by at least  $f+1$ sources}  \label{pgc:freqr1}
\State $c_i=1$
\State $S_{p_s,m_i}=\bot$
\Else
\State $m_i=\bot$
\State $c_i=0$
\State $S_{p_s,m_i}=\bot$
\EndIf

\State \Return{$<m_i, c_i, S_{p_s,m_i}>$}

\end{algorithmic}
\end{algorithm*}
\begin{observation}\label{obs:byz}
Consider an instance of gradecast with source $p_s$. 
If a process $p_i$, whether Byzantine or correct, can produce a {\em seen-all} proof for a message $m$, then each process $p \in C$ delivered message $m$ with rank at least $1$ at the end of the gradecast instance. 
\end{observation}

\begin{proof}
Let us suppose by contradiction that there exists two processes $p_i$ and $p_j$ such that (i) $p_i$ is able to show a {\em seen-all} proof for a message $m$ and (ii) $p_j$ is correct and it never delivered $m$ with rank $1$ or more, during the execution of the instance of provable gradecast algorithm.

If $p_i$ is able to show a {\em seen-all} proof, it means that it collected at least $n-f$ signed copies of $m$ (see line \ref{pgc:proof}). Considering that in the system we have at most $f$ Byzantine processes that can generate fake signed copies of $m$, it follows that the remaining $n-2f$ copies arrive from correct processes.
Since $n \ge 3f +1$, it follows that at least $f+1$ signed copies of $m$ have been generated by correct processes and sent at the beginning of round $3$.
A correct process broadcasts the message to all other processes in the system. This implies that $p_i$ receives at least $f+1$ copies of $m$ at the beginning of round $3$. It remains to show that $m$ is the most frequent message $p_i$ receives. 
Notice that a correct process sends a signed message $m$ at the beginning of round $2$ only upon receiving from a Byzantine quorum of processes (see Line \ref{pgc:sign}), this implies that no two correct processes relay two different messages at the beginning of round 2. 
Therefore, $m$ is the only message that a correct can send to $p_i$, and the Byzantines cannot create more than $f$ messages. This implies that $m$ is the most frequent message that $p_i$ receives and it will be received with frequency at least $f+1$. This contradicts the fact that $p_i$ does not deliver $m$ with rank $1$. \end{proof}

\subsection{Detailed Algorithm Description}

Processes communicate by provable gradecasts. The three rounds necessary to execute a gradecast instance form a single epoch. We assume that in each epoch there are $n$ concurrent instances of gradecast running, one for each possible sender.  The pseudo-code is in Algorithm \ref{gac:algorithm}.

\begin{algorithm*}
\caption{{\em GAC}  - Algorithm for process $p_i$} \label{gac:algorithm}
\footnotesize

\begin{algorithmic}[1]
\algrenewcommand\algorithmicprocedure{\textbf{function}}
\State $t_d=0,t_u=n, t_{m}=\lfloor \frac{t_u}{2} \rfloor$
\State $G=\epsilon$ 
\State $pro_i=\emptyset$ \Comment{proposed value.}
\State $Proofs=\{\}$

\smallskip
\Procedure{LA-Propose}{$pro_i$}
\State $pro_i=pro_i$
\State {\sc gradecast}$(M=(pro_i,G,\bot))$  \Comment{Epoch 0 - start}
\State $P_i=${\sc rcv}()
\State $G=G \cdot s$
\State $updateproofs(pro_i,P_i,0)$   \Comment{Epoch 0 - end}
\For{$ep \in [1,\ldots, \log(n))+1]$} \label{alg:forcycle}
\State {\sc gradecast}$(M=(pro_i,G,Proofs))$ 
\State $P_i=${\sc rcv}()
\State $V_i=${\sc filter}$(P_i)$ \label{alg:masterset}
\If{{\sc classify}$(V_i)=s$}
\State $G=G \cdot s$
\State $t_d=t_d, t_{u}=t_m, t_{m}=t_{d}+\lfloor \frac{t_u-t_d}{2} \rfloor$
\State $updateproofs(pro_i,P_i,ep)$
\Else
\State $G=G \cdot m$
\State $pro_i= V_i$
\State $t_d=t_m,t_{u}=t_u, t_{m}=t_{d}+\lfloor \frac{t_u-t_{d}}{2} \rfloor$
\EndIf

\EndFor
\State {\sc decide}$(\bigoplus pro_i)$ \Comment{The $\bigoplus$ is needed since $V_i$ is a set of sets}
\EndProcedure

\bigskip
\Procedure{classify}{$V$}
\If{ $|V| \leq t_m$} 
\State  {\bf return} $s$
\Else
\State    {\bf return} $m$
\EndIf
\EndProcedure
\medskip
\Procedure{filter}{$P$}
\State $V = \{\forall v \in M | M \in P \land M$ rank is greater than $0\,\, \land$ {\sc admissible}$(v,M) \land \,\, M$ source is in $G \land \,\, v$ in $E\}$
\State {\bf return} $V$
\EndProcedure
\medskip

\Procedure{admissible}{$v,M$} 
\If{$\forall t \in [0,..,|G|-1]$ such that $G[0:t]$ terminates with $s$, there exists an admissibility proof for $v$ in $M$}
\State {\bf return} True
\Else
\State {\bf return} False
\EndIf
\EndProcedure
\medskip

\Procedure{updateproofs}{$pro_i,P_i,ep$}
\ForAll{$v \in pro_i$}
\State $proofv=[\bot]^{ep+1}$
\If{$ep>0$}
\ForAll{$t \in [0,\ldots,ep-1]$}
\If{$G[t]=s$}
\State Let $p$ be a valid proof from $P_i$ for $G[0:t]$ (it must exists since $v$ is admissible). 
\State $proofv[t]=p$
\EndIf
\EndFor
\EndIf
\State Construct a valid proof $p$ for $v$ and group $G$ using messages in $P_i$. \label{alg2:updateproof}
\State $proofv[ep]=p$
\State $Proofs=Proofs \cup \{proofv\}$
\EndFor
\EndProcedure

\end{algorithmic}
\end{algorithm*}

\subsubsection{Epoch $ep =0$}
Epoch $0$ has a special structure. Correct processes belong to a single group $G=\epsilon$. In this epoch all values in $E \subseteq \{\{1\},\ldots,\{n\}\}$ are admissible (i.e., they do not have to carry an admissibility proof), and no classification step is executed.
That is, at the end of the epoch each process goes to group $G=s$ and no correct $p_i$ updates its value $pro_i$. This phase is performed to force Byzantine processes to commit to a certain set of values that cannot be changed later; values that are not associated with a seen-all proof, showing that they have been gradecasted in epoch $0$, will be ignored in the next epochs.

\subsubsection{Epoch $ep \geq 1$}
Epochs $1$ and onward share the same structure. 
At the beginning of an epoch, all correct processes belonging to a same group share the values of three thresholds.  At the beginning of epoch $ep=1$, the group $G=s$ encloses all the  processes and the thresholds are $t_d=0$, $t_m=\frac{n}{2}$ and $t_u=n$.

\paragraph{Value gradecast.} An epoch starts by making each correct process $p_j$ in a group $G$ gradecast a message $M_j$ containing, its proposal value $pro_i$, the group to which $p_j$ belongs, and an {\em admissibility proof} for each element in $pro_i$ (the structure and the precise purpose of this proof is defined later).

Each other correct process $p_i$ in a group $G$  receives, by mean of the gradecast, a set of tuples: $P_i:\{(p_0,M_0,c_0, S_{p_0,M_0}), (p_1,M_{1},c_1, S_{p_1,M_{1}}), \ldots \}$.

We define a special set of values $V_i$ that is a subset of values in messages contained in $P_i$. 
Set $V_i$ contains all ``{\em admissible values}" in $P_i$  and such that: (1) the rank of the message carrying the value is at least $1$ and (2) the sender of the message is in $G$.

Value $v$ is admissible for process $p_i$ in epoch $ep$ if the message that carries $v$ contains also an ``{\em admissibility proof}" proving that $v$ has been seen by all correct processes in $SLV(G[0:1], G[0:2],\ldots,G[0:ep-1])$, where $SLV(Set)$ is a filter function that removes from the set of string $Set$ all the strings that end with letter $m$. As an example considering $G=ssmsm$ we have $SLV(s,ss,ssm,
ssms,ssmsm)=\{s,ss, ssms\}$.   Essentially, there must be a proof showing that $v$ has been seen with rank 1 by all processes in the epochs in which processes in $G$ have been classified as slave.  The actual structure of this proof is described later (Section \ref{sec:proof}). 

Process $p_i$ becomes a group slave if $|V_i| \leq t_m$, it becomes a group master if $t_m < |V_i|$.
If $p_i$ becomes a slave, it enters the group $G \cdot s$. Otherwise, it becomes a master, and it enters $G \cdot m$. 
\medskip

\paragraph{Slaves actions}  If a process $p_i$ is a slave, it updates its set of thresholds as  $t_d=t_d, t_{u}=t_m, t_{m}=t_{d}+\lfloor \frac{t_u-t_d}{2} \rfloor$. Finally, a slave does not update its proposed value $pro_i$, in the next epoch it will have again the exactly same value it had in the current epoch. 
Regarding the admissibility proof, a correct slave has the duty to collect an admissibility proof for its value proposed $pro_i$ this is done by collecting the seen-all proof generated by gradecasting its $pro_i$ at the beginning of the epoch. 

\smallskip
\paragraph{Masters actions}  If process $p_i$ is a master, it updates its set of thresholds as $t_d=t_m,t_{u}=t_u, t_{m}=t_{d}+\lfloor \frac{t_u-t_{d}}{2} \rfloor$. Moreover, it updates its value $pro_i= V_i$.
Regarding the admissibility proof, a correct master has no duty in creating an admissibility proof for its new $pro_i$, but it has to collect proofs to show that any value inserted in $pro_i$ was admissible in $G[0:ep-1]$. 


\subsubsection{Admissibility proof}\label{sec:proof}
A message $m$ containing a value $v$ carries an admissibility proof for $v$ and group $G$ if message $m$ contains for each, also non-proper, prefix $G[0:j]$ of $G$ terminating with character $s$ (for $j$ in $\{0,1,\ldots,|G|-1\}$), a seen-all proof for $m$ with a sender $p$ in $G[0:j-1]$. From Oservation \ref{obs:byz}, it is immediate to see that an admissibility proof for $G$ implies that all correct processes in $SLV(G[0:1], G[0:2],\ldots,G[0:|G|-1])$ received $v$ in the value gradecast phase and ranked the source of the gradecast at least $1$.  See Figure \ref{apf} for a graphical representation of the usefulness of such a proof. 
Assume there exists an admissibility proof for value $v$ and group $G=ssms$, then there is a seen-all proof for the epochs $0,1,3$ and groups in $G_{s}:\{s,ss,ssms\}$(marked as green in the figure).
This implies that, in each of these epochs, value $v$ has been seen by each correct process with rank at least 1. 
In particular, a seen-all proof for value $v$ and group $ss$ implies that $v$ has been gradecasted by a process in group $s$, and it has been received by all correct with rank at least $1$. This implies that $v$ has been inserted in the proposal of all correct masters in group $sm$ (in the figure we represent with an orange border the processes that have $v$ in their proposal). This means that a master in $ssm$ can update its proposal inserting $v$, and it knows that it will still be comparable with the decision of any process in a group with prefix $sm$. 
Iterating the reasoning, a chain of seen-all proofs, the first for group $ss$ and the second for group $ssms$, implies that $v$ is in the proposal of all correct processes in a group with prefix 
$sm$ or $ssmm$, and thus a future master in $ssmmsm$ can safely include $v$ in its proposal. 
The necessity of a seen-all proof for epoch $0$, and thus for group $s$ is needed to force Byzantine processes to commit to at most $f$ values. This is due to the fact that $f$ Byzantine processes are able to create admissibility proofs for at most $f$ distinct values in epoch $0$. 

\begin{figure}
  \centering
    \includegraphics[width=0.4\textwidth]{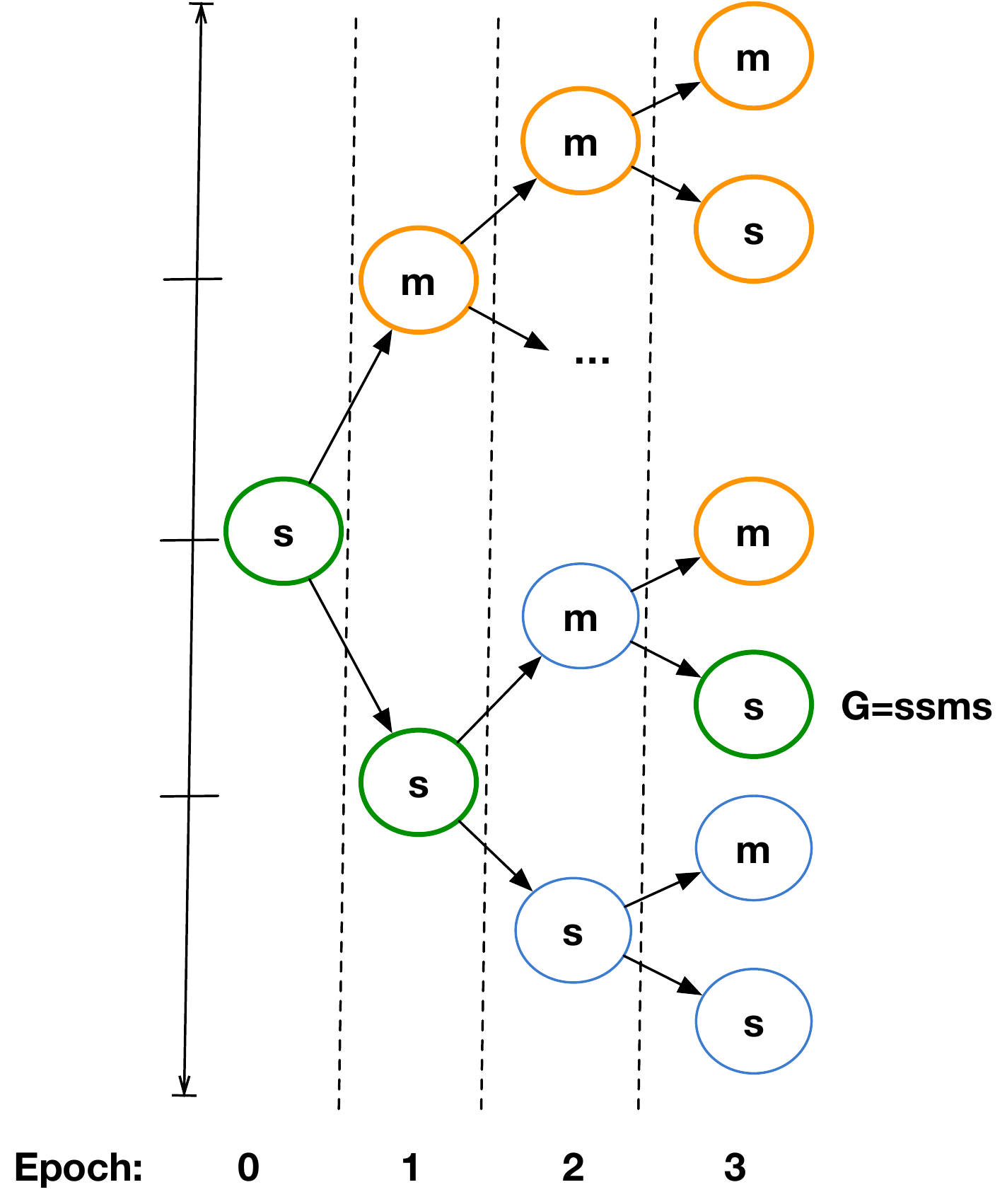}
      \caption{Graphical representation of the purpose of an admissibility proof.  \label{apf} }
\end{figure}

\subsubsection{Termination}
A process $p_i$ terminates the above algorithm when the epoch is $\log(n)+1$. Upon termination it decides its  value $pro_i$.

\subsection{Correctness of GAC}\label{gac:correctness}

\begin{definition} 
Let $A(G)$ be the set of values admissible for correct processes in group $G$ during the gradecast of epoch $|G|$. 
\end{definition}

Given a group $G$ the lemma below shows that the set of admissible values of any other group $G'=G \cdot \sigma$ will be a subset, not necessarily proper, of  $A(G)$.

\begin{Lemma} \label{lemma:nonincreasing}
Let us consider the group $G$ it holds $A(G') \subseteq A(G)$ for any $G' = G \cdot y$. 
\end{Lemma}
\begin{proof}
The proof is by induction on $y$:
	\begin{itemize}
	\item {\bf Base case.} $G'=G\cdot \epsilon$: It is immediate by observing that $G'=G$. Thus $A(G)=A(G')$.
	\item {\bf Inductive case.} We assume the above is true up to $G'=\sigma$, in the inductive step we have to show that it holds for the two possible extensions of $\sigma$.
	Case (1): $G'=\sigma \cdot m$, in such a case the set of admissible values does not change. Thus $A(\sigma)=A(G')$. Case (2): $G'=\sigma \cdot s$, suppose that a value $v$ 
	is in $A(G')$ but not in $A(\sigma)$. In order to be admissible for $G'$ there must exist a proof for each prefix of $G'$ ending with an $s$, this is by construction also an admissibility proof for $\sigma$. Therefore $A(G') \subseteq A(\sigma)$.
	\end{itemize}

\end{proof}

\begin{definition} 
Given an epoch $ep$ we define as $W_{ep}$ the set of values that have been gradecasted by a process in epoch $ep$ and that have been seen by all correct processes with gradecast  rank at least $1$. 
With $W_{ep}(G)$ we indicate the subset of $W_{ep}$ that was sent by processes claiming to belong to group $G$. 
\end{definition}

\begin{Lemma} \label{lemma:masterdominate}
Consider a correct master $p_i$ in the group $G=\sigma \cdot m$. If $p_i$ decides, its decision, let it be $dec_i$, is comparable with the one of any correct slave in $G'= \sigma \cdot s$. 
\end{Lemma}
\begin{proof}
To prove the above it is sufficient to show that each value $v$ in $A(G')$ is included in $dec_i$. In order for a value to be in $A(G')$ it must exist an admissibility proof for group $G'$, the admissibility proof implies two things:
\begin{itemize}
\item (1) that $v$ is in $W_{|\sigma|}(\sigma)$ (see Observation \ref{obs:byz});
\item (2) that $v$ is in $A(\sigma)$.
 \end{itemize}
 From the above, and the code of a correct master $p_i$, $v$ will be in the set $V_i$ of Line \ref{alg:masterset} at epoch $|\sigma|$. Since a master never removes a value from its $pro_i$ its decision $dec_i$ must include $v$. 
 This complete the proof since each correct slave $p$ in $G'$ never put in its proposal a value that is not in $A(G')$: by Lemma \ref{lemma:nonincreasing} the set of values that a slave will consider in any possible future execution of Line \ref{alg:masterset} is
  included in $A(G')$. 
\end{proof}

\begin{Lemma} 
Let $p$ be a correct process that decides $dec_i$. Its decision respects Inclusivity and Non-Triviality. 
\end{Lemma}
\begin{proof}
The inclusivity is immediate from the fact that its proposal $pro_i$ will be admissible for any epoch, thus, in each epoch, it will always be  included in its set $V_i$. 
It remains to show that $dec_i \leq \bigoplus (X \cup B)$. The non-trivial part is to show the bound of $f$ on the set $B$.
We will show by induction on epoch number $ep$ that the sets $B_{ep}$ of admissible values issuable by Byzantines are such that $B_{ep+1} \subseteq B_{ep}$ and that $|B_{0}|=f$.
Note that, $B_{ep+1} \subseteq B_{ep}$ is implied by Lemma \ref{lemma:nonincreasing}.

The bound on $|B_{0}|=f$ derives directly from the gradecast property: a single Byzantine cannot gradecast two different elements of $E$ in epoch $ep=0$. 
\end{proof}

\begin{Lemma} \label{lemma:onceisforever}
Let $p$ be a correct process that at some epoch $ep \leq \log(n)+1$ updates its proposal $pro_i$ with a new value $w$. For any possible $G$ such that $p \in G$ and $\log(n)+1 \geq |G| \geq ep$ we have $w \in A(G)$.  
\end{Lemma}
\begin{proof}
The proof is by induction:
\begin{itemize}
\item Base case: $|G|=ep$: in this case $p$ puts $w$ in its proposal when it belongs to group $G$, since $p$ is correct it does so if and only if $w \in A(G)$. 
\item Inductive step: Let us suppose that it holds for $G$, we will show that it holds also for the two possible extensions $G'=G \cdot s$ and $G''=G \cdot m$. For $G'$ we have that
process  $p$ being correct it updates the admissibility proof for $w$ (line \ref{alg2:updateproof}). This implies that $w \in A(G')$.  Considering $G''$ we have that, by construction of the admissibility proof, $A(G'')=A(G)$ and thus $w \in A(G'')$ (by inductive hypothesis $w \in A(G)$). 
\end{itemize}
\end{proof}

\begin{Lemma} \label{lemma:expdecreasing}
For any correct process $p_i \in G$ with $G \neq \epsilon$ we have $|A(G)| -| pro_i|  \leq n/2^{|G|-1}$. Moreover, given process $p_i \in G$ we have $|A(G)| \leq t_u$ and $pro_i \geq t_d$, where $t_u$ and $t_d$ are the thresholds of group $G$. 

\end{Lemma}
\begin{proof}
We will show that for each $p_i \in G$ it holds that $
 |A(G)| \leq t_u$, that $|pro_i| \geq t_d$, and, that $ t_u - t_d\leq n/2^{|G|-1}$. Recall that $t_u,t_d$ and $t_m$ depend on $G$. The proof is by induction on $G$.
\begin{itemize}
\item Base case: $G=s$: it derives immediately from the structure of the lattice and the fact that $t_d=0$ and $t_u=n$ for all processes in $G$. 
\item Inductive case: By inductive hypothesis the claim holds for $\sigma$. We have $G=\sigma \cdot m$ or $G=\sigma \cdot s$. Let $t_u,t_m,t_d$ be the thresholds of group $\sigma$, by inductive hypothesis we have  $
 |A(\sigma)| \leq t_u$ that $|pro_i| \geq t_d$ and that $ t_u - t_d\leq n/2^{|G|-1}$.
\begin{itemize}
\item Case $G= \sigma \cdot m$.  First notice that for master processes the set of admissible values does not change, that is $A(\sigma)=A(G)$ neither the threshold $t_u$. In such a case we will show that the lattice ``shrinks from below" in the sense
that by updating the lower bound on $pro_i$ our claim holds. Since $p_i$ is in $G$ we have that it updates $pro_i$ and the new $pro_i$ contains a number of elements that are at least $t_{d}+\lfloor \frac{t_{u}-t_{d}}{2} \rfloor+1$, thus the new $t'_{d}=t_d+ \lfloor \frac{t_{u}-t_{d}}{2} \rfloor+1$. 
By immediate algebraic manipulations we have $t_u - t'_d \leq \frac{t_u-t_d}{2}$ that proves our claim.  
\item Case $G= \sigma \cdot s$. In such a case we will show that the lattice ``shrinks from above". Consider the generic $p_i \in G$ this implies that at the end of epoch $|\sigma|$ the set $V_i$ contained at most $t_m$ elements. It is immediate that
since $p_i$ is correct each value in $A(\sigma)$ that is not in $V_i$ cannot be admissible in the extension $G$. Therefore we have $|A(G)| \leq t_m $, now it remains to show that $|pro_i| \geq t_d$ but this is immediate from inductive hypothesis. 
Thus we have $t_u'=t_m$ and $t'_d=t_d$, that implies $t_u'-t_d' \leq t_m - t_d\leq \frac{t_u-t_d}{2}$.  
\end{itemize}
\end{itemize}
\end{proof}

\begin{observation}
Any correct process decides at epoch $\log(n)+1$.
\end{observation}

\begin{Lemma} 
Given any pair $p_i,p_j$ of processes in $C$ their decision $dec_i$ and $dec_j$ are comparable. 
\end{Lemma}
\begin{proof}
Let $G_i$ be the group where $p_i$ belongs at the end of epoch $r=\log(n)+1$, and let $G_j$ the analogous for $p_j$. 
If $G_i$ and $G_j$ share a common prefix $\sigma$ such that $\sigma \cdot m$ is a prefix of $G_i$ and $\sigma \cdot s$ is a prefix of $G_j$, then Lemma \ref{lemma:masterdominate} shows the comparability. \\
The only case when the above (or the symmetric of the above) does not hold is if $G=G_i=G_j$.
In this case, we will show that $dec_i=dec_j$. Suppose the contrary, then we have that there exists at least a value $v \in dec_i$ and such that $v \not\in dec_j$. By Lemma \ref{lemma:onceisforever} we have $v \in A(G)$, and by Lemma \ref{lemma:nonincreasing} $v$ is in each prefixes of $G$. 
Suppose $p_i$ inserted $v$ in its proposal in an epoch $ep$ such that it has been master again in epoch $ep' > ep$, however this implies that also $p_j$ is master in $ep'$, and $p_i$ being  correct in epoch $ep'$ it gradecasts $v$ that will be included in the proposal of $p_j$. 
The above implies that $v$ has been received and inserted by $p_i$ in its proposal exactly in the last epoch in which $p_i$ became master. 

Let $ep_{last}$ be such an epoch. Note that if $t_u-t_{d} \leq 2$ then $v$ is also in the proposal of $p_j$ (both processes enter in the master group in $ep_{last}$):  in case $t_u-t_d=2$ to become master each process has to collect enough values to trespass the threshold $t_m$, recall that $t_m=t_d+1$ and thus it has to collect $t_u$ values, but those and are all the admissible values (by Lemma \ref{lemma:expdecreasing}). In the other case, when $t_u-t_d=1$, a process has also to collect all admissible values ($t_m=t_d$ and $t_u=t_m+1$) (by Lemma \ref{lemma:expdecreasing}).

Therefore,   $t_u -t_{d} > 2$ in $ep_{last}$, however by the structure of the algorithm and the number of epochs being $\log(n)+1$ we eventually have an epoch $ep' > ep_{last}$ such that $t_u-t_{d} =1$ and $t_m=t_d$, when this happens process $p_j$ will become master upon receipt of $v$ from $p_i$ (by Lemma \ref{lemma:onceisforever} $v$ is admissible for $p_j$). This contradicts the fact that $G_i =G_j$, since $p_i$ is never again a master after epoch $ep_{last}$.
\end{proof}

\noindent
From previous lemmas we have:

\begin{theorem}
Given the lattice $L$ constituted by the power set of the first $n$ natural numbers with union as join operation, GAC is a correct LA algorithm on $L$, that terminates in ${\cal O}(\log(n))$ rounds and tolerates up to $n/3-1$ Byzantine processes. 
\end{theorem}

\section{Adapting GAC to work on arbitrary semi-lattices in $\log(f)$ rounds}\label{sec:general}

We first explain how to adapt the GAC algorithm to work in $\log(f)$ on $L_{n}$ when each correct proposes a different unique value in $\{\{1\},\{2\},\ldots, \{n\}\}$. We call such an algorithm GAC$_{fast}$, we then discuss how to adapt  GAC$_{fast}$ to work on a generic join semi-lattice dropping the assumption of different proposal values. 
The main idea is to modify epoch $0$ to satisfy two needs: (1) to force Byzantines to commit to a certain value; (2) to make all processes collect at least $n-f$ different proposal values. 
This allows the thresholds to be set to $t_d=n-f, t_u=n, t_m=\lfloor n-\frac{f}{2} \rfloor$ in all processes at the end of epoch $0$. The modified epoch $0$ is Algorithm \ref{alg:epoch0}.
The code follows the old one with the notable exceptions that: a correct process updates its proposal by including all values, in $E$, that have been seen with rank at least $2$, and it updates its thresholds accordingly.

\begin{algorithm*}
\caption{ {\bf GAC}$_{fast}$: Collect and Commit Epoch $0$  - Algorithm for process $p_i$} \label{alg:epoch0}
\footnotesize

\begin{algorithmic}[1]
\algrenewcommand\algorithmicprocedure{\textbf{function}}

\State $G=\epsilon$ 
\State $pro_i$ \Comment{proposed value.}
\State $Proofs=\{\}$

\smallskip
\Procedure{LA-Propose}{$pro_i$}
\State $pro_i=pro_i$
\State {\sc gradecast}$(M=(pro_i,G,\bot))$  \Comment{Epoch 0 - start}
\State $P_i=${\sc rcv}()
\State $V_i = \{\forall v \in M | M \in P_i \land M$  rank is equal 2 $\land v \in E\}$
\State $G=G \cdot s$
\State $pro_i= V_i$
\State $updateproofs(pro_i,P_i,0)$   \Comment{Epoch 0 - end}
\State $t_d=n-f,t_u=n, t_{m}=t_d+\lfloor \frac{t_u-t_d}{2} \rfloor$
\State $\ldots$   \Comment{Remain as Algorithm \ref{gac:algorithm} but for line \ref{alg:forcycle}}
\EndProcedure
\end{algorithmic}
\end{algorithm*}

The remaining of the algorithm is the same as Algorithm \ref{gac:algorithm} but for line \ref{alg:forcycle} where we have $ep \in [1 \ldots, \log(f)+1]$.

\paragraph{Correctness discussion}
The same lemmas and observations of Section \ref{gac:correctness} hold with the following exceptions: 

\begin{Lemma} \label{lemma:expdecreasingf}
For any correct process $p_i \in G$ with $G \neq \epsilon$ we have $|A(G)| -| pro_i|  \leq f/2^{|G|-1}$. Moreover, given process $p_i \in G$ we have $|A(G)| \leq t_u$ and $pro_i \geq t_d$, where $t_u$ and $t_d$ are the thresholds of group $G$. 
\end{Lemma}
\begin{proof}
We will show that for each $p_i \in G$ it holds $
 |A(G)| \leq t_u$ that $|pro_i| \geq t_d$ and that $ t_u - t_d\leq f/2^{|G|-1}$. Recall that $t_u,t_d$ and $t_m$ depends on $G$. The proof is by induction on $G$.
\begin{itemize}
\item Base case: $G=s$: it derives immediately from the structure of epoch $0$, the property of the gradecast, the fact that there are at least $n-f$ correct processes. At the end of epoch $0$ we have $t_d=n-f$ and $t_u=n$ for all processes in $G$. 
\item Inductive case: By inductive hypothesis the claim holds for $\sigma$. We have $G=\sigma \cdot m/(\sigma \cdot s)$. Let $t_u,t_m,t_d$ be the thresholds of group $\sigma$, by ind. hypothesis we have  $
 |A(\sigma)| \leq t_u$ that $|pro_i| \geq t_d$ and that $ t_u - t_d\leq f/2^{|G|-1}$.
\begin{itemize}
\item Case $G= \sigma \cdot m$.  First notice that for master processes the set of admissible values does not change, that is $A(\sigma)=A(G)$ neither the threshold $t_u$. In such a case we will show that the lattice ``shrinks from below" in the sense
that by updating the lower bound on $pro_i$ our claim holds. Since $p_i$ is in $G$ we have that it updates $pro_i$ and the new $pro_i$ contains a number of elements that are at least $t_{d}+\lfloor \frac{t_{u}-t_{d}}{2} \rfloor+1$, thus the new $t'_{d}=t_d+ \lfloor \frac{t_{u}-t_{d}}{2} \rfloor+1$. 
By immediate algebraic manipulations we have $t_u - t'_d \leq \frac{t_u-t_d}{2}$ that proves our claim.  
\item Case $G= \sigma \cdot s$. In such a case we will show that the lattice ``shrinks from above". Consider the generic $p_i \in G$ this implies that at the end of epoch $|\sigma|$ the set $V_i$ contained at most $t_m$ elements. It is immediate that
since $p_i$ is correct each value in $A(\sigma)$ that is not in $V_i$ cannot be admissible in the extension $G$. Therefore we have $|A(G)| \leq t_m $, now it remains to show that $|pro_i| \geq t_d$ but this is immediate from inductive hypothesis. 
Thus we have $t_u'=t_m$ and $t'_d=t_d$, that implies $t_u'-t_d' \leq t_m - t_d\leq \frac{t_u-t_d}{2}$.  
\end{itemize}
\end{itemize}
\end{proof}

\begin{observation}
Any correct process decides at epoch $\log(f)+1$.
\end{observation}

\begin{theorem}
Given the lattice $L$ constituted by the power set of the first $n$ natural numbers with union as join operation and where each correct process proposes a distinct element in $\{\{1\},\ldots, \{n\}\}$, GAC$_{fast}$ is a correct LA that terminates in ${\cal O}(\log(f))$ rounds and tolerates up to $n/3-1$ Byzantine processes. 
\end{theorem}

\subsection{Arbitrary Semi-lattices $L_{A}$}
We adapt GAC$_{fast}$ to work on  an arbitrary join semi-lattice $L_{A}=(V_{A},\oplus)$, an arbitrary set $E_{A} \subseteq V_{A}$ of allowed proposal values, and an arbitrary mapping of proposal values and correct processes (recall that in previous section we were assuming a different proposal value for each correct). The adaptation works by running GAC$_{fast}$ on an intermediate semi-lattice $L^{*}$. Lattice $L^{*}$ is the one induced by the union operation over the power set of $V^{*}=\Pi \times E_{A}$. The set $V^{*}$ is constituted by all possible pairs process ID and initial proposed value $pro_i$ (each $pro_i$ is in $E_{A}$). Each correct process $p_i$ starts GAC$_{fast}$ with input $(p_i, pro_i)$.
\begin{algorithm*}
\caption{ ${\cal O}(\log f)$ Lattice Agreement Algorithm for arbitrary join semi-lattice - Algorithm for process $p_i$} \label{alg:final}
\footnotesize

\begin{algorithmic}[1]
\algrenewcommand\algorithmicprocedure{\textbf{function}}

\medskip
\Procedure{\bf LA-Propose}{$pro_i$}
\State Trigger GAC$_{fast}$ with proposal value $(p_i,pro_i)$
\State wait until GAC$_{fast}$ terminates with output $dec_i$
\State $X:\{y | (x,y) \in dec_i \land y \in E_{A}\}$
\State $dec'_i=\bigoplus X$
\State {\bf return} $dec'_i$
\EndProcedure

\end{algorithmic}
\end{algorithm*}

Note that epoch $0$, with is commitment functionality, forces the algorithm to effectively decides on a lattice $L^{*}$ that is the power set of a subset $X$ of $V^{*}$ of cardinality at most $n$ and at least $n-f$. 
Such lattice is isomorphic to the lattice on which the correctness of GAC$_{fast}$ has been shown in Section \ref{sec:general}. 

Once GAC$_{fast}$ terminates each process $p_i \in C$ has a decision $dec_i$. This decision $dec_i$ is a set $\{ (p_i,val_{i}), \ldots, (p_x,val_{x}) \}$, from such a set the process $p_i$ obtains a decision $dec'_i$ on $L_{A}$ where $dec'_i$ is $dec'_i = \bigoplus D_i$ given $D_i: \{y | \forall (x,y) \in dec_i  \land y \in E_{A}\}$. This strategy enforces that the decisions $dec'_i$ and $dec'_j$ of any two correct processes $p_i,p_j$ are comparable points on the semi-lattice $L_{A}$.  It is also immediate that non-triviality and inclusivity hold.  The formal pseudocode of the reduction is Algorithm \ref{alg:final}.
From the above we have:
 \begin{theorem} Given $f$ Byzantine processes and $n$ processes in total,
if $n \geq 3f+1$, there exists a Byzantine lattice agreement algorithm terminating in ${\cal O}(\log f)$ rounds in the authenticated message model. 
\end{theorem}

\subsubsection{Message Complexity}
The provable gradecast generates at most ${\cal O}(n^2)$ messages at each round. Each epoch is composed by $3$ rounds and in each epoch all correct processes do a gradecast, thus we have a total of ${\cal O}(n^3 \log f)$ messages. However, as noted also in \cite{Zheng:aa}, it is possible to use ${\cal O}(n^2)$ messages in total to run $n$ parallel instances of gradecast (each message will be structured with $n$ locations one for each possible gradecaster).   Therefore,  our algorithm can be implemented using ${\cal O}(n^2 \log f)$ messages.

\section{Trade-off Between Signatures and Number of Processes}\label{sec:nosig}
 
Signatures are used to implement the seen-all proof of our provable gradecast primitive (explained in Section \ref{pgc:sec}). By assuming $n \geq 4f+1$ processes we may implement an interactive version of the seen-all proof that does not use signatures. We explain the interactive provable gradecast algorithm when $n=4f+1$, the extension for $n > 4f+1$ is immediate. The modifications with respect to Section \ref{pgc:sec} are as follows:
\begin{itemize}
\item The threshold of line \ref{pgc:freqr1} remains $f+1$, the one of line \ref{pgc:freqr2} becomes $3f+1$, the threshold of line \ref{pgc:relay} is $3f+1$. Therefore, a message will have rank $1$ if seen with multiplicity at least $f+1$, a message has rank $2$ if seen with multiplicity at least $3f+1$. 
\item The seen-all proof $S_{p_s,m_i}$ is simply a set of IDs. These IDs are the ones of processes from which $p_i$ receives $m_i$ in the receive phase of round 3. 
\end{itemize}

The seen-all proof $S_{p_s,m_i}$  is checked in an interactive way by querying each process contained in the set $S_{p_s,m_i}$. The proof passes if at least $2f+1$ of such processes confirm to have relayed $m_i$ in the relevant gradecast instance. It is obvious that by increasing the two thresholds we do not affect the original properties of the gradecast (discussed in Section \ref{pgc:sec}).
 
\begin{observation}\label{obs:byznosig}
Consider an instance of interactive provable gradecast with source $p_s$. 
If a process $p_i$, whether Byzantine or correct, can produce a passable {\em seen-all} proof for a message $m$, then each process $p \in C$ delivered message $m$ with rank at least $1$ at the end of the gradecast instance. 
\end{observation}
\begin{proof}  In the interactive provable gradecast if a proof produced by process $p_i$ is passable, then it implies that at least a set of $f+1$ correct processes have relayed message $m$ to $p_i$ at the end of round $3$ of the gradecast.  Being these processes correct they will relay $m$ to all other processes in the system. Therefore, any other process has received $m$ with frequency at least $f+1$ at the end of round $3$. This implies that each correct process delivers $m$ with rank at least $1$.  \end{proof}

Moreover, we have the following observation:
\begin{observation}\label{obs:byzcorrnosig} 
Consider an instance of interactive provable gradecast with source $p_s$. 
Given a correct process $p_i$, if $p_i$ receives a message $m$ with rank $2$ then the corresponding proof is passable. 
\end{observation}

Using the interactive provable gradecast and the straightforward interactive variant of the admissibility proof we have: 
 \begin{theorem} Given $f$ Byzantine processes and $n$ processes in total, 
if $n \geq 4f+1$, then there exists a byzantine lattice agreement algorithm terminating in ${\cal O}(\log f)$ rounds in the authenticated channel model. 
\end{theorem}

\paragraph{Message Complexity.}
We argue that the interactive provable gradecast generates at most ${\cal O}(n^2)$ messages at each round. The additional cost introduced by the interactive proof is at most ${\cal O}(n)$ per round: this is the number of messages that a process has to send to verify the admissibility of all values inside a set (we have to send at most one query message to each other process independently of how many values we want to check). Thus the total asymptotic cost remains the same: ${\cal O}(n^2 \log f)$ messages.

\section{An universal transformer from LA to Generalised LA }\label{bzgc:const}

 \cite{Faleiro:2012} introduced the Generalised Lattice Agreement problem. Such problem is essentially a version of LA in which processes are allowed to propose and decide on a, possibly infinite, sequence of values. The equivalent problem for Byzantine failures, in asynchronous system, has been defined in \cite{ipdps}. In this section we show a transformer algorithm that builds upon a generic ``one-shot" LA algorithm to create a Byzantine tolerant Generalised LA algorithm. 
 
\paragraph{Synchronous Generalised LA.} We consider the definition of \cite{ipdps} adapted for a synchronous system. 
In the Synchronous Generalised LA, each correct process $p_i$  receives input values from an infinite 
sequence $Pro_i=\langle pro_{0},pro_{1}, pro_{2},\ldots \rangle$ where each $pro_k$ is a value inside a set of admissible values $E$ (note that $E$ is not necessarily finite). Without loss of generality we imagine that at each round $r$, $p_i$ receives a value $pro_r \in Pro_i$ (note that this is not restrictive since we could modify the lattice to admit a neutral element, such as $\emptyset$).  A correct process $p_i$ must output an infinite number of decision values $Dec_i = \langle dec_{0},dec_{1}, dec_{2},\ldots \rangle$. 
The sequence of decisions has to satisfy the following properties: 

 \begin{itemize}
\item {\bf Liveness:} each correct process $p_i \in C$ performs an infinite sequence of decisions $Dec_{i}= \langle dec_{0},dec_{1}, dec_{2},\ldots \rangle$;
\item {\bf Local Stability:} For each $p_i \in C$  its sequence of decisions is non decreasing (i.e., $dec_{h} \subseteq dec_{h+1}$, for any $dec_{h} \in Dec_{i}$);
\item {\bf Comparability:} Any two decisions of correct processes are comparable, even when they happen on different processes; 
\item {\bf Inclusivity:} Given any correct process $p_i \in C$, if $Pro_i$ contains a value $pro_k$, then $pro_k$ is eventually included in $dec_{h} \in Dec_{i}$;
\item {\bf Non-Triviality:}  Given any correct process $p_i \in C$ if $p_i$ outputs some decision $dec_k$ at a round $r$, then 
 \\ $dec_k \leq \bigoplus (Prop[0:r] \cup B[0:g(r)])$, where, $Prop[0:r]$  is the union of the prefixes, until index $r$, of all sequences $Pro_i$ of correct processes; and, $B[0:g(r)]$ is the union of all prefixes, until index $r$,  of 
 $f$  infinite sequences $B_i$, one for each Byzantine process. Function $g$ is $g: \mathbb{N} \rightarrow \mathbb{N}$. Each $B_i$ is a sequence of elements in $E$. 
\end{itemize}

Intuitively, function $g$ upper bounds, to $f \cdot g(k)$ the number of values that Byzantine processes can insert up to decision $k$. 

\subsection{Transformer}
We now explain the high level idea behind the transformer. 
Let ${\cal LA}$ be a one shot synchronous lattice agreement algorithm that terminates in $\delta$ rounds. 
We divide the time in terms, a term lasts for $\delta$ rounds and it allows to execute, from start to termination, an instance of ${\cal LA}$. 
At the beginning of term $k$, correct processes start the $k$-th instance of ${\cal LA}$, we denote it as  $k$-${\cal LA}$. Each correct process receives from upper layer a stream of elements in $E$, and it batches such elements until a new instance of ${\cal LA}$ starts. Let $C_{k}$ be the $k$-th batch, at the  beginning of term $k$, process $p_i$ starts the instance $k$-${\cal LA}$ with input $(p_i, dec_{k-1} \oplus C_{k})$, where $dec_{k-1}$ is the output of the $(k-1)$-${\cal LA}$ instance.  
There are few minor details to add to this description to get an actual algorithm (see pseudocode in Algorithm \ref{trans:algorithm}), the most important being the mechanism needed to bound the number of values that could be added by Byzantine processes (recall that we have to upper bound their values by a function $g$). The key idea is to start the instance $k$-${\cal LA}$ with a set of admissible values $ P_{T(k-1)+\delta}$, that is the set of all subsets of $E$ of size less or equal to $T(k-1)+\delta$.  Function $T$ is a function mapping each index of the decision sequence of a correct process to an upper bound on the maximum size of the decision, where the size is counted as number of elements in the decision. We assume $T(-1)=0$; we have that $T(0)=\delta \cdot n$: each correct process that starts the first instance of LA proposes at most $\delta$ values; the closed form for $k>0$ of $T(k)$ is in the statement of Lemma \ref{lemma:transf}.

\paragraph{Correctness discussion.}

Assuming that each  ${\cal LA}$ is an instance of a correct LA algorithm (according to definition in Section \ref{definition:la}), we argue that the Generalised LA algorithm obtained by the transformer described above is a correct algorithm according to definition in Section \ref{bzgc:const}. 
The liveness property is satisfied by the liveness of each instance of LA. The local stability and the inclusivity derive directly from the fact that once a process $p$ outputs $dec_{k-1}$, the next instance of LA will have as input a value that contains $(p, dec_{k-1} \oplus C_{k})$. Therefore, by inclusivity of LA we have that the decision of  $k$-${\cal LA}$ contains the pair $(p, dec_{k-1} \oplus C_{k})$, and this means that $ dec_{k-1} \oplus C_{k} \leq dec_{k}$. It remains to show the non-triviality:

\begin{Lemma}\label{lemma:transf}
Consider the sequence of decisions of a correct process $p$ executing the transformer in Algorithm \ref{trans:algorithm}. This sequence respects the Non-Triviality property for a function $g(r)< (T(k)=\frac{\delta \cdot n ((f+1)^{k+1}-1)}{f})$
\end{Lemma}
\begin{proof}
We have to upper bound the number of different values of $E$ that could be inserted by byzantines in the decision $dec_k$ of a correct process. Without loss of generality we suppose that $E$ is infinite. 
We will upper bound the total number $T(k)$ of elements in $E$ that could be in a decision $dec_k$ ($dec_k$ is the $k$-th elements of the sequence of decisions $Dec_i = \langle dec_{0},dec_{1}, dec_{2},\ldots \rangle$). Clearly, this upper bound is an upper bound on $g(k)$.
When $k=0$ we have that each process is allowed to propose a $(id,Y)$ where $|Y|=\delta$, by the Non-Triviality of LA we have $|dec_{0}| \leq n\delta$. When $k > 0$ we have the following recurrence relation on $T$ that is $T(k) \leq f \cdot T(k-1) +T(k-1) + n \delta$, where: the term $(n -f)\delta$ accounts for the new $\delta$ elements that each correct process may add (the $C$ sets); the term $T(k-1)$ account for the fact that each correct process proposes also its previous decision (that has maximum size $T(k-1)$); and finally the term $f \cdot T(k-1)+ f \delta$ accounts for the fact that each byzantine could propose a set of new elements of size $T(k-1)+\delta$. Essentially, a Byzantine could fake its decision $dec_{k-1}$ to be any subset of $E$ of size at most $T(k-1)$.
 \end{proof}


\begin{algorithm*}
\caption{{\em Transformer}  - Algorithm for process $p_i$} \label{trans:algorithm}
\footnotesize

\begin{algorithmic}[1]
\algrenewcommand\algorithmicprocedure{\textbf{upon event}}
\State $C=\emptyset$ \Comment{Batch of values}
\State $dec=\emptyset$ \Comment{Decision}

\Procedure{propose}{$pro_r$}

\State $C=C \cup \{pro_r\}$

\EndProcedure

\Procedure{ $r = \delta \cdot (k+1)$ for some $k \in \mathbb{N}$}{}
\State Start the instance $k$ of $LA$ over lattice ${\cal L}_{k}$ with admissible values $E_{k}=\Pi \times P_{T(k-1)+\delta}$ 
\State {\sc $k$-LA-Propose}$((p_i, dec \oplus C))$
\State $C=\emptyset$
\EndProcedure

\Procedure{decision from the current instance of LA}{$dec'$}
\State $X:\{y | (x,y) \in dec'\}$
\State $dec=\bigoplus X$ 
\State Decision$_k(dec)$
\EndProcedure

\end{algorithmic}
\end{algorithm*}

\bibliographystyle{plain}
\bibliography{bibliography.bib}

\begin{thebibliography}{10}

\bibitem{Attiya:1995}
Hagit Attiya, Maurice Herlihy, and Ophir Rachman.
\newblock Atomic snapshots using lattice agreement.
\newblock {\em Distributed Computing}, 8(3), 1995.

\bibitem{doi:10.1002/0471478210}
Hagit Attiya and Jennifer Welch.
\newblock {\em Distributed Computing}.
\newblock John Wiley \& Sons, Ltd, 2004.

\bibitem{Ben-or:2010}
Mickael Ben-Or, Danny Dolev, and Ezra Hoch.
\newblock Simple gradecast based algorithm.
\newblock Technical Report arXiv:1007.1049, arXiv preprint, 2010.

\bibitem{doty}
Ho-Lin Chen, Rachel Cummings, David Doty, and David Soloveichik.
\newblock Speed faults in computation by chemical reaction networks.
\newblock {\em Distrib. Comput.}, 30(5):373--390, October 2017.

\bibitem{ipdps}
Giuseppe~Antonio {Di Luna}, Emmanuelle Anceaume, and Leonardo Querzoni.
\newblock Byzantine generalized lattice agreement.
\newblock In {\em Proceedings of the 34th IEEE International Parallel and
  Distributed Processing Symposium (to appear)}, 2020.
\newblock https://arxiv.org/abs/1910.05768.

\bibitem{popprot}
Giuseppe~Antonio {Di Luna}, Paola Flocchini, Taisuke Izumi, Tomoko Izumi,
  Nicola Santoro, and Giovanni Viglietta.
\newblock Population protocols with faulty interactions: The impact of a
  leader.
\newblock {\em Theor. Comput. Sci.}, 754:35--49, 2019.

\bibitem{icdcn}
Giuseppe~Antonio {Di Luna}, Paola Flocchini, Giuseppe Prencipe, Nicola Santoro,
  and Giovanni Viglietta.
\newblock Line recovery by programmable particles.
\newblock In {\em Proceedings of the 19th International Conference on
  Distributed Computing and Networking, {ICDCN} 2018, Varanasi, India, January
  4-7, 2018}, pages 4:1--4:10, 2018.

\bibitem{Faleiro:2012}
Jose~M. Faleiro, Sriram Rajamani, Kaushik Rajan, G.~Ramalingam, and Kapil
  Vaswani.
\newblock Generalized lattice agreement.
\newblock In {\em Proceedings of the ACM Symposium on Principles of Distributed
  Computing (PODC)}, 2012.

\bibitem{105555}
Maurice Herlihy, Dmitry Kozlov, and Sergio Rajsbaum.
\newblock {\em Distributed Computing Through Combinatorial Topology}.
\newblock Morgan Kaufmann Publishers Inc., San Francisco, CA, USA, 1st edition,
  2013.

\bibitem{10.5555/2821576}
Nancy~A. Lynch.
\newblock {\em Distributed Algorithms}.
\newblock Morgan Kaufmann Publishers Inc., San Francisco, CA, USA, 1996.

\bibitem{michail}
Othon Michail, Paul~G. Spirakis, and Michail Theofilatos.
\newblock Fault tolerant network constructors.
\newblock {\em CoRR}, abs/1903.05992, 2019.

\bibitem{Nation}
J.~B. Nation.
\newblock Notes on lattice theory.
\newblock \url{http://math.hawaii.edu/~jb/math618/Nation-LatticeTheory.pdf}.

\bibitem{Nowak:2019}
Thomas Nowak and Joel Rybicki.
\newblock {Byzantine Approximate Agreement on Graphs}.
\newblock In Jukka Suomela, editor, {\em 33rd International Symposium on
  Distributed Computing (DISC 2019)}, volume 146 of {\em Leibniz International
  Proceedings in Informatics (LIPIcs)}, pages 29:1--29:17, Dagstuhl, Germany,
  2019. Schloss Dagstuhl--Leibniz-Zentrum fuer Informatik.

\bibitem{6812814}
M.~{Raynal}.
\newblock {\em Fault-tolerant Agreement in Synchronous Message-passing
  Systems}.
\newblock Morgan \& Claypool, 2010.

\bibitem{Skrzypczak}
Jan Skrzypczak, Florian Schintke, and Thorsten Sch{\"u}tt.
\newblock Linearizable state machine replication of state-based crdts without
  logs.
\newblock In {\em Proceedings of the 2019 ACM Symposium on Principles of
  Distributed Computing - PODC '19}. ACM Press, 2019.

\bibitem{Zheng:aa}
Xiong Zheng and Vijay~K. Garg.
\newblock Byzantine lattice agreement in synchronous systems.
\newblock {\em CoRR}, abs/1910.14141, 2019.
\newblock http://arxiv.org/abs/1910.14141.

\bibitem{Garg:2018b}
Xiong Zheng, Vijay~K. Garg, and John Kaippallimalil.
\newblock Linearizable replicated state machines with lattice agreement.
\newblock https://arxiv.org/abs/1810.05871, 2018.
\newblock arXiv:1810.05871.

\bibitem{Garg:2018}
Xiong Zheng, Changyong Hu, and Vijay~K. Garg.
\newblock Lattice agreement in message passing systems.
\newblock In {\em Proceedings of the International Symposium on Distributed
  Computing (DISC)}, 2018.

\end{thebibliography}

\end{document}